\documentclass[conference]{IEEEtran}

\ifCLASSINFOpdf
   \usepackage[pdftex]{graphicx}
\else
\fi

\usepackage{amsmath, amsthm, amssymb}
\usepackage{comment}
\usepackage{verbatim}
\usepackage{setspace}
\usepackage{graphicx}
\newtheorem*{thma}{Lemma}

\begin{document}
%
\title{Improving the GPU space of computation under triangular domain problems.}

\author{\IEEEauthorblockN{Crist\'obal A. Navarro}
\IEEEauthorblockA{Computer Science department\\
University of Chile\\
Chile\\
Email: crinavar@dcc.uchile.cl}
\and
\IEEEauthorblockN{Nancy Hitschfeld}
\IEEEauthorblockA{Computer Science department\\
University of Chile\\
Email: nancy@dcc.uchile.cl}}

\maketitle
\begin{abstract}
There is a stage in the GPU computing pipeline where a grid of thread-blocks is mapped to the problem domain. 
Normally, this grid is a $k$-dimensional bounding box that covers a $k$-dimensional problem no matter its shape. Threads that 
fall inside the problem domain perform computations, otherwise they are discarded at runtime.
For problems with non-square geometry, this is not always the best idea because part of the space of computation is executed without any practical use.
Two-dimensional triangular domain problems, alias \textit{td-problems}, are a particular case of interest. 
Problems such as the Euclidean distance map, LU decomposition, collision detection and simulations over triangular tiled domains are all \textit{td-problems} and they appear frequently in many 
areas of science. 
In this work, we propose an improved GPU \textit{mapping function} $g(\lambda)$, that maps any $\lambda$ block to a unique location $(i,j)$ in the triangular domain.
The mapping is based on the properties of the lower triangular matrix and it works at a block level, thus not compromising thread organization within a block. 
The theoretical improvement from using $g(\lambda)$ is upper bounded as $I < 2$ and the number of wasted blocks is reduced from $\mathcal{O}(n^2)$ to $\mathcal{O}(n)$. 
We compare our strategy with other proposed methods; the \textit{upper-triangular mapping} (UTM), the \textit{rectangular box} (RB) and the \textit{recursive partition} (REC). 
Our experimental results on Nvidia's Kepler GPU architecture show that $g(\lambda)$ is between 12\% and 15\% faster than the bounding box (BB) strategy. When compared to the other strategies, 
our mapping runs significantly faster than UTM and it is as fast as RB in practical use, with the advantage that thread organization is not compromised, as in RB. This work also contributes at presenting, 
for the first time, a fair comparison of all existing strategies running the same experiments under the same hardware.
\end{abstract}

\IEEEpeerreviewmaketitle

\section{Introduction and Related work}
\label{sec_introduction}
GPU computing is without question an important research area \cite{citeulike:2767438, Nickolls:2010:GCE:1803935.1804055} since the release 
of general purpose computing platforms such as CUDA \cite{nvidia_cuda_guide} and OpenCL \cite{opencl08}. 
For every GPU application, there is a stage where a grid of blocks, also known as the \textit{space of computation}, 
is mapped to a problem domain to eventually solve it.
This mapping can be defined as a function $f(x): R^k \rightarrow R^p$ that transforms each $k$-dimensional point 
$x=(x_1, x_2, ..., x_k)$ of the grid to a unique $p$-dimensional point of the problem domain. 
In other words, $f(x)$ \textit{maps the space of computation to the problem domain}. 

When the problem domain is simple in shape, rectangular or square grids are good choices because the bounding box matches exactly the domain. 
Rectangular or square grids are the most used ones and they are characterized for using the bounding box strategy (BB), where $f(x) = x$.

There are also other types of problems that do not match a box shaped domain because they have a different geometry. 
More in detail, some 2D problems have a triangular shaped domain. We call these type of problems \textit{triangular-domain-problems} 
or simply \textit{td-problems}. Building a square grid for a \textit{td-problem} is not the best choice 
because entire blocks of computation, containing hundreds of threads, are wasted and discarded at runtime, leading to a cost in 
performance. The scenario is illustrated in Figure \ref{fig_bb_strategy}.
\begin{figure}[ht!]
\centering
\includegraphics[scale=0.14]{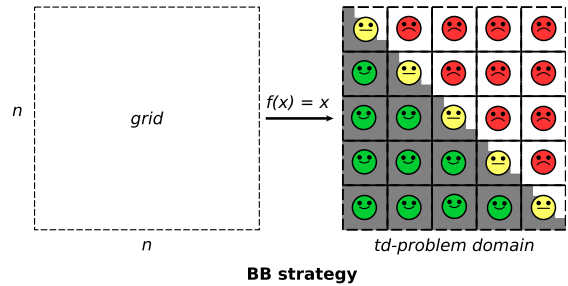}
\caption{The BB strategy is not the best choice for \textit{td-problem}.}
\label{fig_bb_strategy}
\end{figure}

Many computational problems are in fact \textit{td-problems}; 
simulations over triangular tilings (e.g John Conway's Game of life \cite{ConwaysLife}), collision detection tables, LU decomposition, 
graph adjacencies and Euclidean distance maps (EDM), among others. The problem of improving the space of computation is important, 
because each contribution in this matter will have an impact on every \textit{td-problem}.
Related work in the field of distance maps has proposed GPU implementations for parallel computation of DNA sequence distances \cite{Ying:2011:GDD:2065356.2065583} which 
is based on EDM. In their work, Ying et al. mention that the problem domain is indeed symmetric and they do realize that only 
the upper or lower triangular part of the interaction matrix requires computation. Li et al. \cite{Li:2010:CME:1955604.1956601}
have also worked on GPU-based EDMs on large data and have also identified the symmetry involved in the computation. 
In both works, there has not been a proposal for a strategy regarding the mapping of the grid to the problem domain. 
Packed data structures have been proposed in the field of lineal algebra for representing triangular and symmetric matrices 
with applications to LU and Cholesky decomposition \cite{springerlink_gustavson}. Jung \textit{et. al.} \cite{Jung2008} proposed the \textit{rectangular box strategy} 
for accessing and storing a triangular matrix (upper of lower). 
As a result, they achieve data structures with practically half the size with respect to classical methods based on the full matrix. 
In 2009, Ries et al. contributed with a parallel GPU method for the triangular matrix inversion \cite{Ries:2009:TMI:1654059.1654069}. 
The authors identify that the space of computation can be improved by using a \textit{recursive partition} of the grid, based on \textit{divide and conquer}. 
In 2012, Q. Avril \textit{et. al.} \cite{AvrilGA12} proposed 
a mapping function that given a thread id $k$, it computes the coordinates $(a, b)$, based on the properties of 
the \textit{upper-triangular} section of a symmetric matrix. The authors mention that they use Carmack's and Lomont's 
fast square root approximation (based on the Newton-Raphson approximation algorithm \cite{Peelle:1974:TNS:585882.585889}) for speeding up the mapping function. The authors 
also mention that all approximation errors can be fixed by using only two conditionals statements. 
The motivation behind our work is similar to the one of Q. Avril \textit{et. al.}, but it uses the idea of the \textit{lower-triangular matrix} instead of the upper one, 
and instead of mapping threads, we map blocks. These two differences are critical for achieving a simpler, 
faster and exact mapping function for $N < 30,000$. 

Up to date, there has not been a dedicated comparison of the different strategies proposed for improving the space of computation. In the best case, we can find a comparison against the BB strategy 
\cite{AvrilGA12} but the authors did not give details if the BB kernel was optimized or not. For example, 
filtering by block coordinates whenever is possible is faster than filtering by the thread id. 
In this paper we address this lack of comparisons and present for the first time results for all strategies running the same tests under the same hardware.

The rest of the manuscript includes a formal definition of $g(\lambda)$ (section \ref{sec_td_strategy}), details about its implementation and how we chose the best square root function 
(section \ref{sec_implementation}). In section \ref{sec_experimental_results} we present experimental results for 
all existing strategies. All strategies are tested under the same conditions; to execute a dummy kernel and a kernel for computing 
the EDM using $1$, $2$, $3$ and $4$ features. Both kernels run in the range $N \in [1024, 30720]$ with $N$ a multiple of $1024$. 
Results, advantages and disadvantages are discussed in section \ref{sec_conclusions}.

\section{The mapping function}
\label{sec_td_strategy}
\subsection{Formulation}
\label{sec_block_mapping_function}
Let $A$ be a \textit{td-problem} of size $N(N+1)/2$,  $n = \lceil N/\rho\rceil$ the number 
of blocks needed to cover the data along a dimension and $\rho$ the number of threads per block per dimension, or 
\textit{dimensional blocksize}. 
A BB strategy would simply build a square grid, namely $G_{BB}$, of $n$x$n$ blocks and put conditional instructions to cancel the computations 
outside the problem domain. A finer analysis tells that $n(n+1)/2$ blocks are sufficient to cover the problem domain of $A$. 
These blocks can be indexed in the following way:
\begin{equation}
A = 
\begin{vmatrix}
0	&		&		&				&					\\
1	&	2	&		&				&					\\
3	&	4	&	5	&				&					\\	
...	&	... 	&	...	&	...			& 					\\
\frac{n(n-1)}{2}	&	\frac{n(n-1)}{2} + 1 	&	...	&	...	& 	\frac{n(n+1)}{2}-1\\
\end{vmatrix}
\label{eq_matrix_indexing}
\end{equation}
The idea is to first build a two-dimensional balanced grid, namely $G_{LTM}$, that will contain all $n(n+1)/2$ blocks. 
By balanced, we mean that the size per dimension of the grid must be $n' = \lceil\sqrt{n(n+1)/2}\rceil$ per dimension. 
For the rest of the paper we will name our method as LTM for \textit{lower triangular mapping}.
Figure \ref{fig_ltm_strategy} illustrates $G_{LTM}$ and how it is smaller than $G_{BB}$ (from Figure \ref{fig_bb_strategy}) 
while still providing the necessary amount of blocks to cover the problem domain.
\begin{figure}[ht!]
\centering
\includegraphics[scale=0.14]{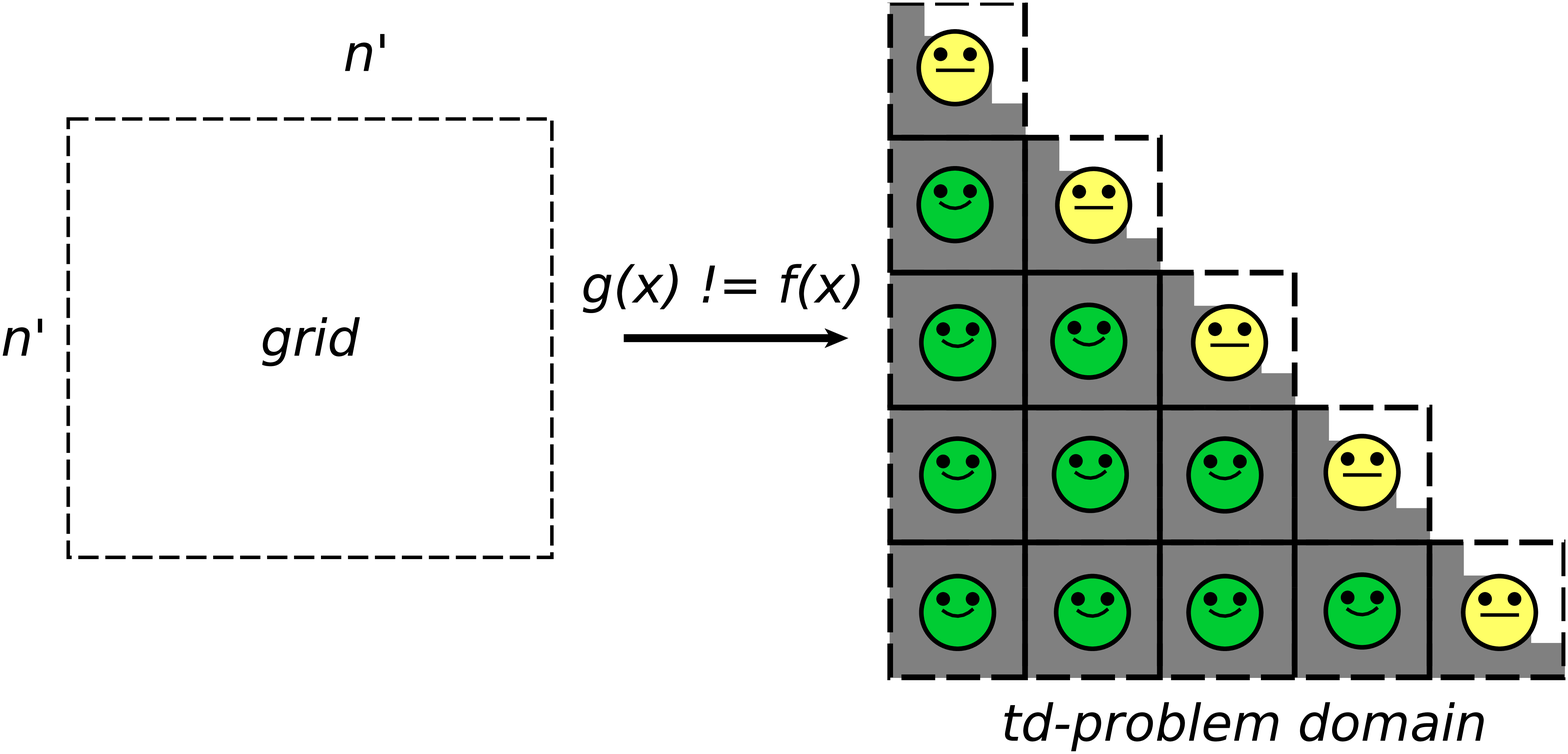}
\caption{The LTM strategy uses just the necessary amount of blocks to cover the problem domain.}
\label{fig_ltm_strategy}
\end{figure}
The result is a reduction from $n(n-1)/2 \in \mathcal{O}(n^2)$ to $n/2 \in \mathcal{O}(n)$ wasted blocks. In other words, conditional 
statements will only occur when the block lies on the diagonal, in order to filter the threads in the upper part.

The next step is to formulate $g(\lambda) = (i,j)$ where $(i,j)$ are the coordinates in problem space and 
$\lambda$ is the index of the block $B_{x,y}$ computed as $\lambda = x + yn'$ in grid space. 
\begin{thma}
	For any block $B_{x,y}$ with index $\lambda = x + yn'$ and $\lambda \in [0, n(n+1)/2-1]$, 
	its corresponding $i,j$ coordinates are computed with a mapping function $g(\lambda)$:
	\begin{align}
	g(\lambda) = (i,j) = \Big(\Big\lfloor\sqrt{\frac{1}{4} + 2\lambda} - \frac{1}{2}\Big\rfloor, \lambda - i(i+1)/2\Big)
	\label{eq_theorem}
	\end{align}
\end{thma}
\begin{proof}[Proof]
	Block index $\lambda$ can be written as:
	\begin{equation}
	\lambda = \sum_{k=1}^{x}k
	\label{eq_packed_adyacency_sum_c}
	\end{equation}
	The index of the far left block that lies on the same row of the $\lambda$-th block corresponds to the sum in the range $[1,i]$. 
	Similarly, the index of the far left block of the $(i+1)$-th row is 
	a sum in the range $[1,i+1]$. That is, for all $\lambda$ indices under the same row $i$, 
	the range of the summation will lie in a range $[1,i+\epsilon]$, where $\epsilon < 1$. With this observation, $\lambda$ can be bounded:
	\begin{align}
	\sum_{k=1}^{i}k & \le \lambda < \sum_{k=1}^{i+1}k
	\label{eq_packed_adyacency_sum}
	\end{align} 
	Using (\ref{eq_packed_adyacency_sum_c}) in (\ref{eq_packed_adyacency_sum}), we get that $x \in [i,i+1)$ therefore $i = \lfloor x\rfloor$.\\
	Equation (\ref{eq_packed_adyacency_sum_c}) can be re-written as
	\begin{equation}
	\lambda = \frac{x(x+1)}{2} 
	\label{eq_packed_adyacency_sum_c_2}
	\end{equation}
	which is a second order equation with $a=1$, $b=1$ and $c=-2$:
	\begin{equation}
	x^2 + x - 2\lambda = 0
	\label{eq_packed_adyacency_equation}
	\end{equation}
	with a positive solution of:
	\begin{equation}
	x = \frac{\sqrt{1 + 8\lambda} - 1}{2} = \sqrt{1/4 + 2\lambda} - 1/2
	\label{eq_packed_adyacency_x}
	\end{equation}
	The row of the $\lambda$-th block can now be computed as:
	\begin{equation}
	i = \lfloor x \rfloor = \Big\lfloor \sqrt{1/4 + 2\lambda} - 1/2 \Big\rfloor
	\end{equation} 
	Finally, $j$ is the distance from the $\lambda$-th block to the left most block in the same row:
	\begin{equation}
	j = \lambda - \frac{i(i+1)}{2}
	\label{eq_packed_adyacency_j}
	\end{equation}
\end{proof}

If the diagonal is not needed, then $g(\lambda)$ becomes:
\begin{align}
g(\lambda) = (i,j) = \Big(\Big\lfloor\sqrt{\frac{1}{4} + 2\lambda} + \frac{1}{2}\Big\rfloor, \lambda - i(i+1)/2\Big)
\label{eq_theorem_nodiag}
\end{align}

When comparing LTM and UTM \cite{AvrilGA12}, we identify several differences: (1) LTM is based on the lower triangular-matrix mapping. (2) $g(\lambda)$ uses fewer floating point operations than in UTM. 
(3) LTM maps blocks and not threads as in UTM. (4) Since $g(\lambda)$ is a map of blocks and the number of blocks is $n = N/\rho$, the square root gives smaller approximation errors, allowing 
exact computation for larger values of $N$. 

\subsection{Bounds on the improvement factor}
The LTM strategy depends strongly on the square root which is asymptotically $O(M(n))$ \cite{Ypma:1995:HDN:222504.222510} where $M(n)$ is the cost of multiplying 
two numbers of $n$ digits. 
Considering that real numbers are represented by a finite number of digits 
(i.e floating point numbers with a maximum of $m$ digits), then all basic operations cost a fixed amount of
time units, leading to a constant cost $M(m) = C_s \in \mathcal{O}(1)$. All other computations are elemental arithmetic operations and 
will be taken as an additional cost of $C_a \in \mathcal{O}(1)$. The LTM strategy has a cost of 
$\tau = C_s + C_a = \mathcal{O}(1)$ for each mapping performed. 
On the other hand, the BB strategy checks for each block if $B_y <= B_x$ in order 
to know if the threads inside have to do work or not, leading to a constant cost of $\beta \in \mathcal{O}(1)$. 
For this particular case, asymptotic analysis will not give useful information about the improvement factor, since both the 
LTM and BB strategies lie in the same complexity order (i.e $\mathcal{O}(n^2)$). Therefore, 
we proceed with a finer analysis in order to find the constants involved in the improvement factor.

Let $|G_{BB}|$ and $|G_{LTM}|$ be the amount of blocks for the BB and LTM strategies, respectively, and $\rho$ the amount of threads 
per block per dimension as mentioned earlier. If we assume that $\beta$ is cheaper than $\tau$, we get that $\tau = k\beta$ with a 
constant $k \ge 1$. The improvement factor $I$ can be obtained by dividing the total cost of BB by LTM across their entire 
grids:
\begin{equation}
I = \frac{\beta|G_{BB}|\rho^2}{\tau|G_{LTM}|\rho^2} = \frac{\beta n^2}{\tau n(n+1)/2} = \frac{2\beta n^2}{\tau n^2 + \tau n}
\label{eq_I}
\end{equation}
As shown in (\ref{eq_I}), the improvement does not depend on the threads, but instead, on the blocks. For large $n$, $I$ becomes:
\begin{equation}
I = \frac{2\beta n^2}{\tau n^2 + \tau n} \approx \frac{2\beta}{\tau}
\label{eq_I_bounded}
\end{equation}
If $\tau \ge 2\beta$, performance is equal to or worse than that of BB. 
A real improvement is achieved when:
\begin{equation}
\beta \le \tau < 2\beta
\label{eq_tau_range}
\end{equation}
By using the relation $\tau = k\beta$ in (\ref{eq_I_bounded}) we get that:
\begin{equation}
I \approx 2/k
\end{equation}
Since $k > 1$, the range of $I$ is:
\begin{equation}
0 < I < 2
\end{equation}
Any value in the range $1< I < 2$ means an improvement in performance and a value in the range $0 < I < 1$ will mean a 
slowdown respect to the BB strategy.
Constant $k$ can be interpreted as the cost and overhead of the mapping function. 
A value of $k\approx1$ means that the maximum possible improvement is achieved; $I_{max} \approx 2$, under large $n$.
In practice, a value of $k\approx1$ is too optimistic and will not occur in practice. Our hypothesis is that actual hardware could give a value in the range 
$1.5 \le k < 2.0$ which would correspond to $1.00 < I \le 1.33$. Any value of $k \ge 2$ will lead 
to no improvement at all, resulting in slower performance than the BB strategy. 
It is important to put emphasis on the fact that $C_a$ (arithmetic operations) will not have much room for optimization as $C_s$. 
Therefore, getting the maximum possible value of $I$ will finally depend on how small is $C_s$, which is the square root.

\section{Implementation}
\label{sec_implementation}
\subsection{Choosing the best implementation for LTM}
The performance of the LTM strategy depends strongly on how fast the computation of index $i$ is. More precisely, the 
computation of the square root as mentioned earlier. 
For this, we made three implementations of the LTM strategy and tested them against the BB strategy in order to keep the fastest one. 
The first one, named 'LTM-X', uses the default $sqrtf(x)$ function from CUDA; this is the simplest one. 

The second implementation of LTM, named 'LTM-N', computes the square root by using three iterations of the Newton-Raphson method 
\cite{Ypma:1995:HDN:222504.222510, Peelle:1974:TNS:585882.585889}. More in detail, we use the implementation of Carmack and Lomont. 
This implementation became famous because it has proved to be effective for applications that allow small errors.
The initial value used is the magic number '0x5f3759df' (this initial guess became known when 'Id Software' released Quake 3 source code back in the year 2005). 
We added a constant of $\epsilon = 10^{-4}$ to the computation of $i$ to automatically repair the floating point point error. 
With this small change, the computation of the $i$ coordinate becomes exact for the range $N \in [0, 30720]$.

The third implementation, named 'LTM-R', uses the hardware implemented reciprocal square root, $rsqrtf(x)$:
\begin{equation}
\sqrt{x} = \frac{x}{\sqrt{x}} = x\ rsqrtf(x) 
\end{equation}
This implementation is as simple as LTM-X, with the only difference that it adds $\epsilon = 10^{-4}$ just like in LTM-N.

We measured the improvement factor of each implementation, running a dummy kernel that only computes the $i,j$ indices and writes 
the sum $i+j$ to a constant location in memory. It is necessary to perform a memory write using the coordinates, otherwise the compiler 
can optimize the code removing part of the mapping cost. 
Figure \ref{fig_results_map} shows the improvement factor between BB and LTM as well as a comparison on the amount of wasted blocks.
\begin{figure*}[ht!]
\centering
\begin{tabular}{cc}
\includegraphics[scale=0.7]{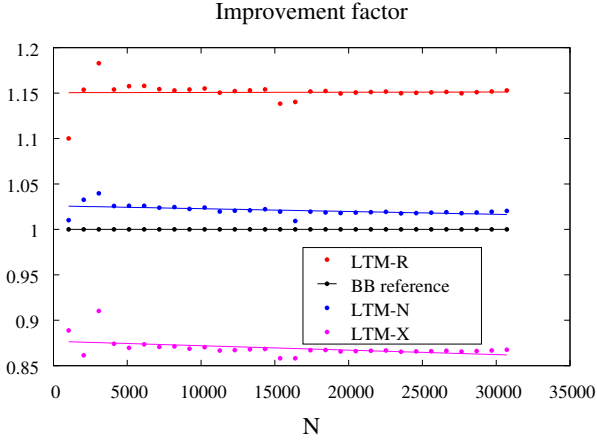} & \includegraphics[scale=0.7]{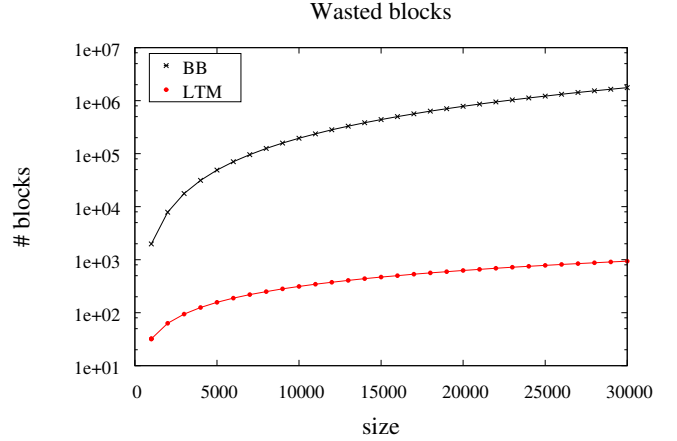}
\end{tabular}
\caption{Only the LTM-R implementation runs faster than the BB strategy. On the right, the number of wasted blocks a a function of $N$.}
\label{fig_results_map}
\end{figure*}

From the results, we observe that LTM-X is slower than BB, only achieving $I \approx 0.7$. LTM-N achieves an improvement of $I \approx 1.03$ which is only a little better than BB.
LTM-R achieved a value of $I\approx1.15$. From the results we observe that using the inverse square root is the best option, thus we keep LTM-R.

We also tried an alternative version of LTM using a lookup table for the mapping, but performance was as slow as using 
the default square root $sqrtf(x)$, even when the Kepler architecture has an efficient mechanism for accessing a common value 
for all threads (uniform-load feature) within a block. We did not include these results because memory limitations did not allow us to 
experiment in the full range of $N \in [1024, 30720]$. It is important to mention that the implementation of the LTM strategy worked fast for 
the Kepler architecture but not for Fermi (previous architecture), which achieved $I < 1$ in all cases. 
\subsection{Implementing the other strategies}
Apart from LTM, we also implemented BB, RB, REC and UTM strategies following the details given by the authors. However, we added the following restriction: 
the mapping cannot use the global memory of the GPU. This means no auxiliary array and no lookup tables. A couple of constants are allowed though. 
We made this restriction because we are considering practical use of these strategies, that is, to dedicate the entire GPU memory for the \textit{td-problem} being solved.

For the \textit{bounding box} (BB) strategy, we make blocks above the diagonal to be discarded immediately, without needing to compute a thread coordinate. This is done by checking the following: 
if $B_x > B_y$ is true, then the thread returns. For the rest of the threads, the coordinate is computed. The condition $i > j$ is still performed to discard threads above the diagonal, where $B_x = B_y$. 
This implementation of BB is faster than computing the thread coordinate for every block and filtering afterwards.

The \textit{rectangular box} (RB) strategy is based on the work of Jung \textit{et. al.} \cite{Jung2008}. This method takes the sub-triangular portion of the threads where $t_x > N/2$, rotates it CCW and 
places it above the diagonal to form a vertical rectangular grid (see Figure \ref{fig_other_methods}). In the original work, the authors map the thread coordinates with the help of a texture. In this case, 
we perform the coordinate mapping arithmetically without using any texture. All threads below the diagonal just need to do $i = t_y-1$, while $j$ remains the same. 
Threads on or above the diagonal must compute $i = N - t_y - 1$ and $j = N - i - 1$. This mapping is for even values of $N$. The case of odd $N$ is the same idea but partitioning at $\lfloor N/2 \rfloor$.

The \textit{recursive partition} (REC) strategy was proposed for the GPU by \cite{Ries:2009:TMI:1654059.1654069} \textit{et. al.}. In this method, the size of the problem is defined 
as $N = m2^k$ where $k$ and $m$ are positive integers and $m$ is a multiple of $\rho$ (the blocksize). The idea is do a binary \textit{bottom-up} recursion of $k$ levels, similar to \textit{merge-sort} (see Figure \ref{fig_other_methods}). 
At each level, a grid of blocks is launched for parallel execution. Their method requires an additional pass for computing the blocks at the diagonal. The details of how the grid is built and how blocks are distributed are well explained in 
\cite{Ries:2009:TMI:1654059.1654069}. 
In the original work, the mapping of blocks to their respective locations at each level is achieved by using a lookup table stored in constant memory. 
In this work, we do the mapping at runtime as with RB. 
\begin{figure}[ht!]
\centering
\includegraphics[scale=0.18]{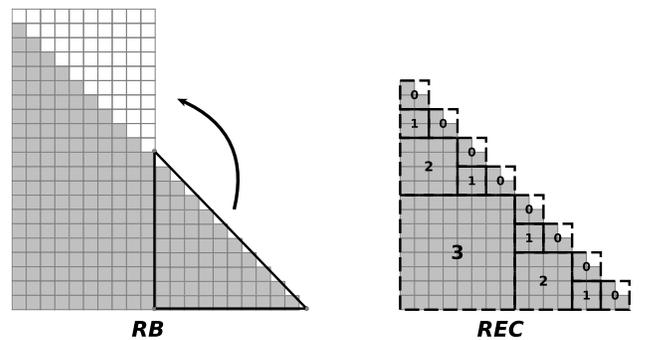}
\caption{The RB and REC strategies.}
\label{fig_other_methods}
\end{figure}

The \textit{upper-triangular mapping} (UTM) was proposed by Avril \textit{et. al.} \cite{AvrilGA12} for performing efficient collision detection on the GPU. 
This method is very similar to LTM. Given a problem size $N$, 
and a thread index $k$, its unique pair $(a, b)$ is given by $a = \lfloor \frac{-(2n + 1) + \sqrt{4n^2 - 4n - 8k +1}}{-2} \rfloor$ and $b = (a+1) + k - \frac{(a-1)(2n-a)}{2}$. This strategy uses the idea 
of mapping to the upper-triangular matrix without the diagonal. Mapping to the upper-triangular matrix still allows solving lower-triangular shaped domains, and \textit{vice-versa} via transposition.

\section{Experimental results}
\label{sec_experimental_results}
Our experimental design consists of measuring the performance of LTM-R and compare it against all existing strategies; 
\textit{bounding box} (BB), \textit{upper-triangular mapping} (UTM) \cite{AvrilGA12}, \textit{rectangular box} (RB) \cite{Jung2008} and 
the \textit{recursive partition} (REC) \cite{Ries:2009:TMI:1654059.1654069}. We checked the output for each strategy to be always correct and the same, for all cases.
Two tests are performed to each strategy; (1) the dummy kernel and (2) the EDM kernel. The dummy kernel just adds the coordinates 
and saves the result into a fixed memory location, the implementation is as simple as possible. 
The purpose of the dummy kernel is to measure just the cost of the strategy and not the problem being solved. 
Test (2) consists of a real problem; to compute the Euclidean distance matrix (EDM) using one, two, three, and four features. The purpose 
of the second test is to measure what is the performance of all strategies when solving a real problem. Testing with different number of 
features will give an idea on what is the behavior when increasing the work per thread.

The hardware used for all experiments is shown in table (\ref{table_hardware}).
\begin{table}[ht!]
\caption{Hardware used for experiments.}
\begin{center}
\begin{tabular}{|c|l|}
\hline
Component	&	Description\\ \hline
CPU		&	Intel(R) Core(TM) i7-3770K @ 3.50GHz \\
RAM		&	32GB DDR3 1333MHZ\\
GPU		&	Geforce GTX 680 (2GB, 1536 cores) \\
API		&	CUDA 5.0 \\ \hline
\end{tabular}
\end{center}
\label{table_hardware}
\end{table}

The results for the dummy kernel, EDM-1D, EDM-4D and the improvement behavior are shown in Figure \ref{fig_all_results}.
\begin{figure*}[ht!]
\centering
\begin{tabular}{cc}
\includegraphics[scale=0.7]{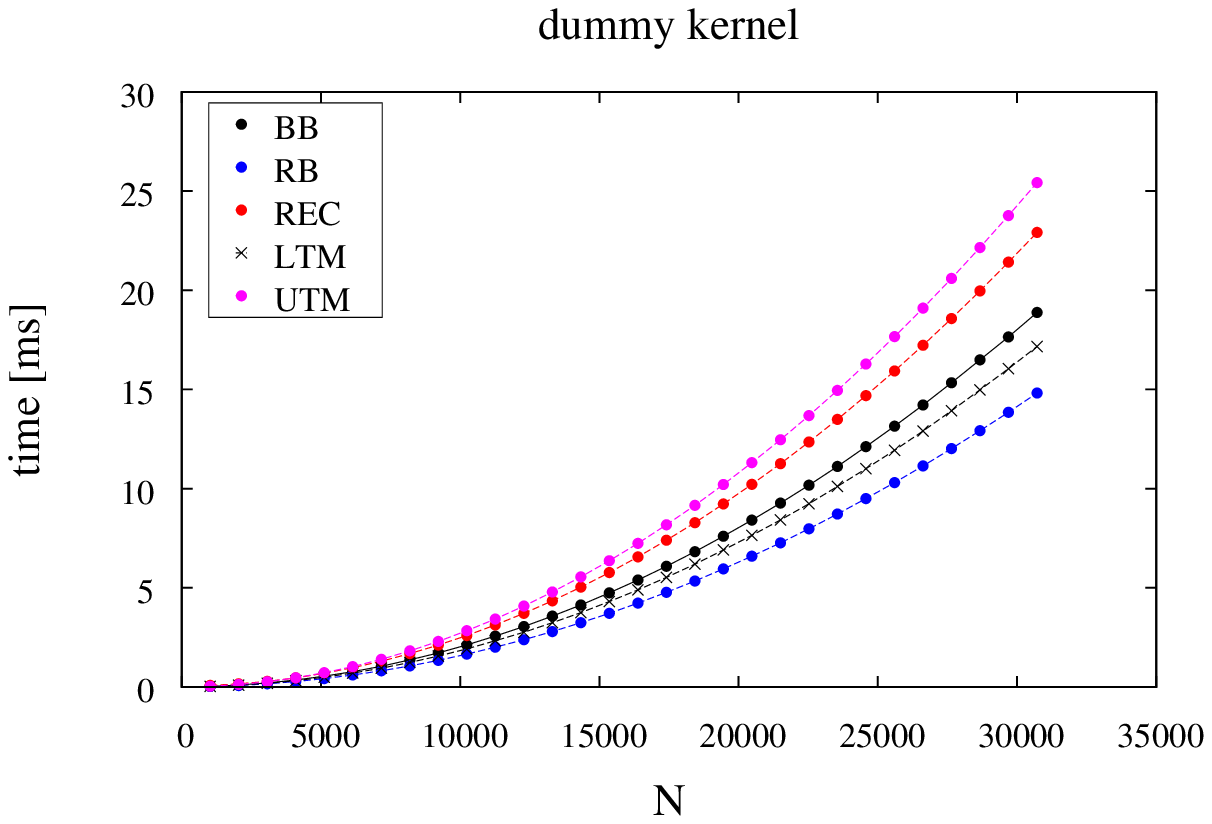} & \includegraphics[scale=0.7]{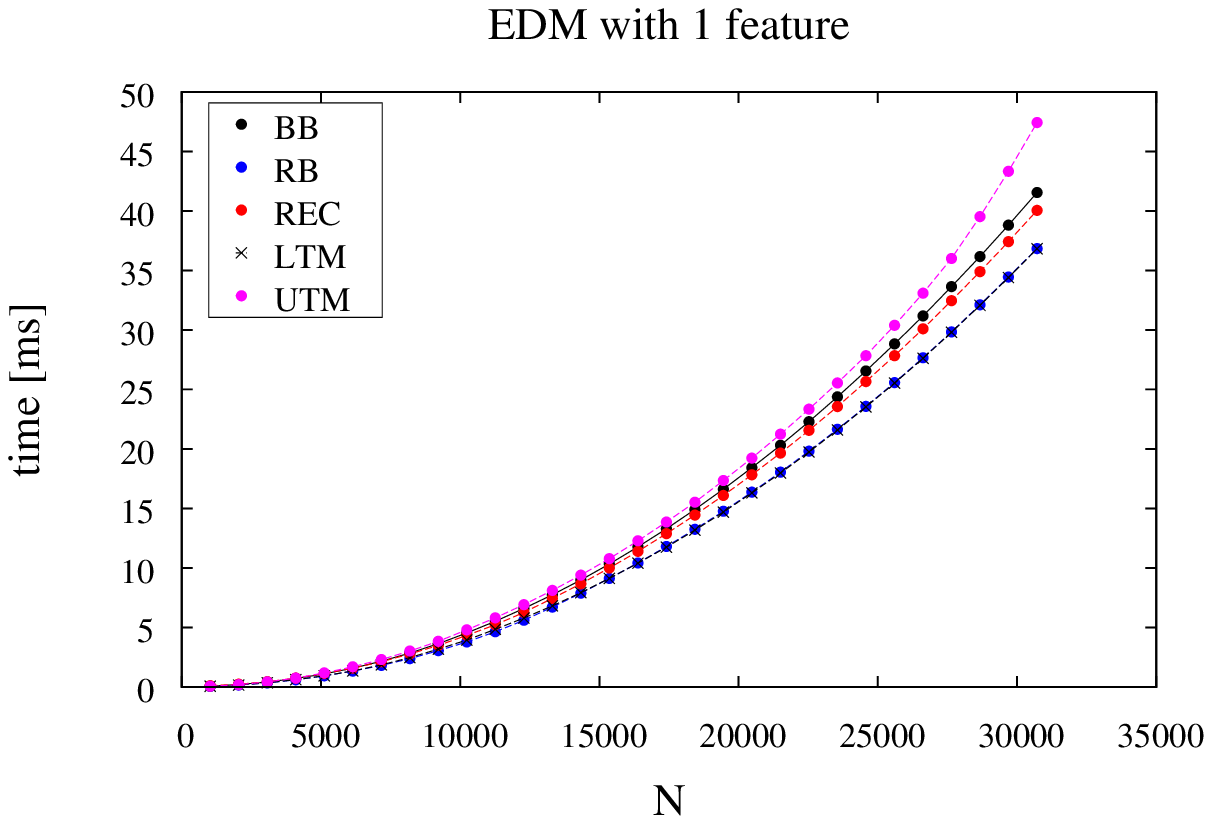}\\
\includegraphics[scale=0.7]{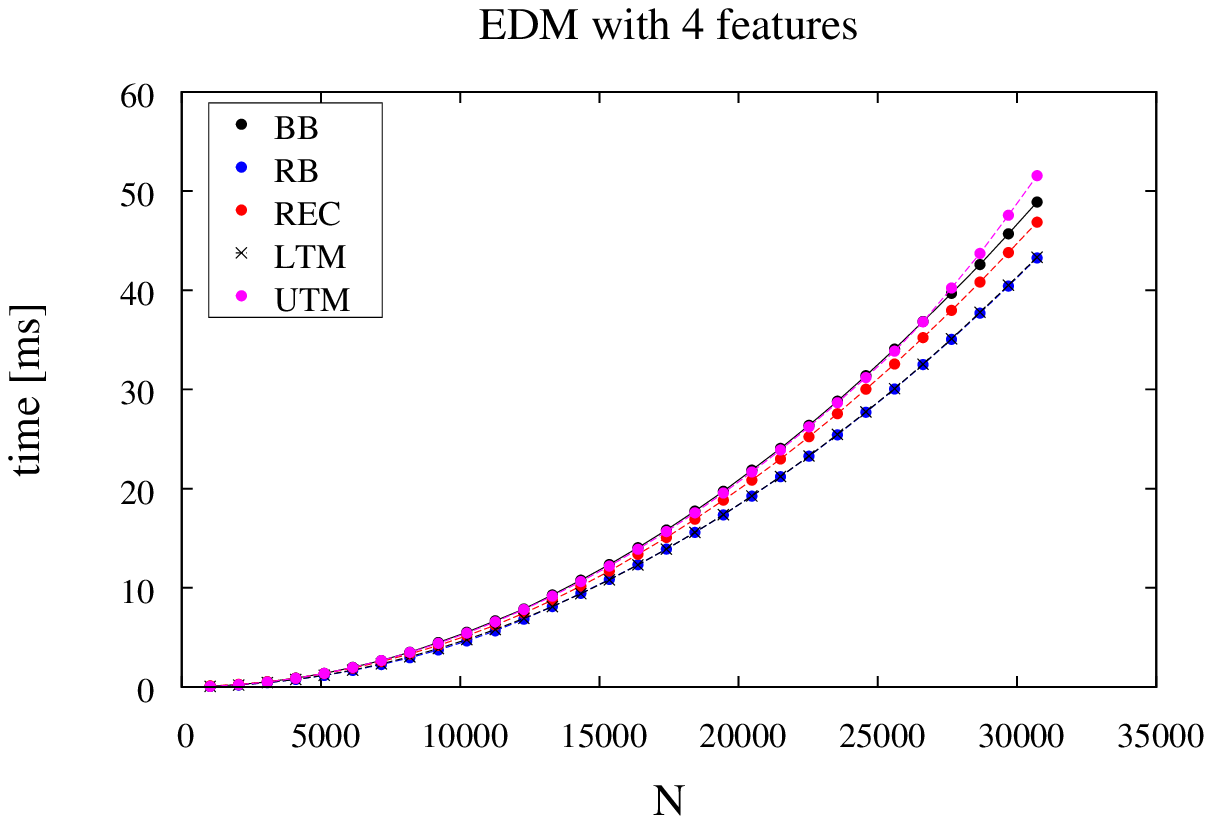} & \includegraphics[scale=0.7]{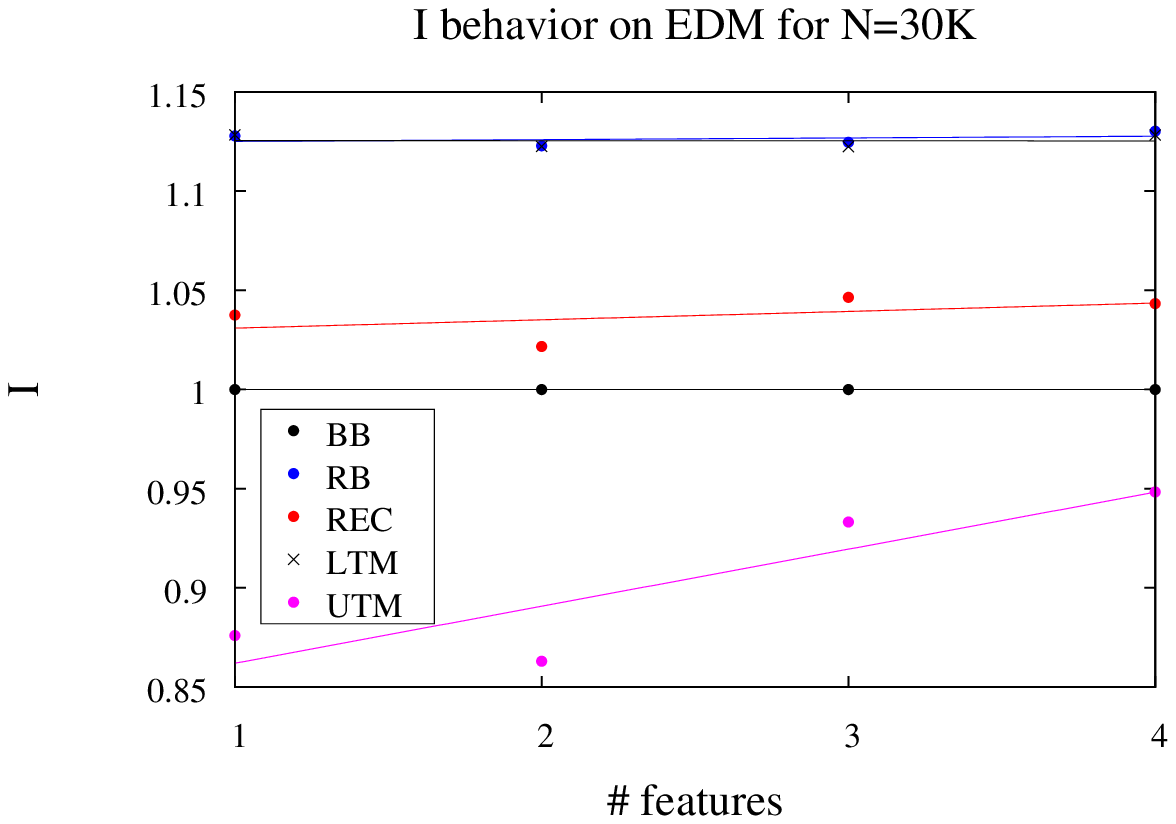}
\end{tabular}
\caption{Results for the dummy kernel, EDM and the behavior of $I$ when scaling the number of features.}
\label{fig_all_results}
\end{figure*}

Results on the dummy kernel show that the RB strategy is the fastest one achieving up to 25\% of improvement over BB. LTM comes in the second place achieving approximately 13\% of improvement over BB. 
The REC and UTM strategies gave an unexpected result; they performed slower than BB for the whole range of $N$.

The EDM problem is solved by computing the Euclidean distance $d_{ij}$ for all pairs. The distance between a pair of elements 
$a_i$, $a_j$ is computed as:
\begin{equation}
f(i,j) = \sqrt{\sum_{k=1}^{d}({a_i^k} - {a_j^k})^2}
\end{equation}
Where $d$ is the number of features and superscript $k$ specifies which feature of the element will be used.

From the results on EDM, we observe that the improvement of RB is reduced to an average of 14\%, practically half of what it achieved with the dummy kernel. 
LTM maintains the same performance as in the dummy kernel, within the range 12\% \~ 15\%. The performance of REC strategy now becomes faster than BB for $N>5,000$ and achieves up to 5\% of improvement 
for $N=30,720$. 
UTM achieves the slowest performance of all strategies. The behavior of UTM can only be explained by the the overhead incurred in computing the square root and repairing the coordinates 
with conditionals. We confirmed in section \ref{sec_implementation} that manual square root computation is still not fast enough. 
Results for EDM using two and three features showed similar results, therefore we decided not to include them.

In the last figure, we can see that the behavior of $I$ as a function of the number of features with $N=30,720$ does not vary significantly for LTM and RB. 
For the case of REC and UTM, the number of features does have an impact on performance, being beneficial for UTM.

\section{Conclusion}
\label{sec_conclusions}
Our main result from this work is the proposal of a function $g(\lambda)$ for mapping a grid-blocks to a any \textit{td-problem}, 
achieving an average improvement factor of $I=1.15$ with respect to the BB strategy, where the theoretical range is $0 < I < 2$. 
The reason for such improvement is the fact that the number of wasted blocks at runtime is reduced from 
$\mathcal{O}(n^2)$ to $\mathcal{O}(n)$, making the cores waste less time on unnecessary conditional instructions. 
The implementation of the LTM strategy is extremely short in code and totally detached from the problem, making it easy to adopt.
Such improvement of $I = 1.15$ was only achievable when using the inverse square root (LTM-R). We think that this technical aspect is key 
to understanding that the square root plays a critical role in the cost of the mapping function. For our case, we found that the reciprocal square root was the most convenient option. 
The $\epsilon$ value worked correctly in the range that fits on our GPU memory, that is $N \in [$1 to 30,000$]$ (using blocks of size 16x16). 
This range is already useful for practical use. For bigger ranges, $\epsilon$ must be refined or another approach must be taken 
for fixing eventual errors. As long as the error is $e \le 1$, block-level (\textit{i.e.}, non-branching) conditionals can fix the result.
When comparing our method with other state of the art strategies, we found that LTM and RB are the fastest methods up to date. Additionally, LTM does not compromise thread organization as RB, giving an 
important advantage. This means that when using shared memory and thread coarsening, the implementation will become easier using LTM rather than RB, since latter would need to 
put conditionals in order to run different code when processing threads below, on and above the diagonal.

The results obtained for REC are still promising. Even with the overhead of $k$ kernel calls, it still performs faster than BB. We think that this strategy can become even faster for GPU architectures that 
allow recursive kernel invocation. 

The results of UTM were not expected to be slower than BB. Our hypothesis is that the cause of such performance is the manual computation of the square root and the 
conditionals involved for repairing the computed value. We believe that if UTM uses the block mapping approach and the inverse square root trick, its performance will increase considerably 
and will not need conditionals for the range $N \in [1,30720]$.

In the future, GPUs will eventually become faster, each time having more 
\textit{special function units} (SFU) and FP32 units per multi-processor, speeding up the computation of the square root. 
At that point, the LTM strategy will be able to use the default $sqrt$ function and achieve an improvement factor of $I > 1$. 
As a final conclusion, we can say that improving the space of computation for \textit{td-problem} has proven to be advantageous in theory as well as in practice.

\section*{Acknowledgment}
The authors would like to thank \textit{Anonymous} for funding the PhD program of the first author. 
This work was partially supported by the project XXXX.



\bibliographystyle{IEEEtran}
\bibliography{td_hipc}
%


\end{document}